\documentclass[a4paper,11pt]{article}
\pdfoutput=1 

\usepackage{jheppub} 

\usepackage[utf8]{inputenc} 

\usepackage{amsmath}
\usepackage{amssymb}
\usepackage{amsthm}
\usepackage{mathrsfs}
\usepackage{comment}

\usepackage{physics}
\usepackage{braket}

\newcommand{\N}{\mathbb{N}}
\newcommand{\Z}{\mathbb{Z}}
\newcommand{\C}{\mathbb{C}}
\newcommand{\R}{\mathbb{R}}

\newtheorem{proposition}{Proposition}

\newcommand{\defeq}{:=}

\usepackage{xcolor}

\title{Laughlin states change under large geometry deformations and imaginary time Hamiltonian dynamics}

\author[a]{Gabriel Matos}
\author[c,d]{Bruno Mera}
\author[b,d]{José M. Mourão}
\author[e]{Paulo D. Mourão}
\author[b,d]{João P. Nunes}

\affiliation[a]{
School of Physics and Astronomy, University of Leeds, Leeds LS2 9JT, UK
}

\affiliation[b]{Center for Mathematical Analysis, Geometry and Dynamical Systems, Instituto Superior Técnico, Universidade de Lisboa, 1049-001 Lisboa, Portugal} 	
\affiliation[c]{Instituto de Telecomunicações, 1049-001 Lisboa, Portugal}
\affiliation[d]{Department of Mathematics, Instituto Superior Técnico, Universidade de Lisboa, 1049-001 Lisboa, Portugal} 	
\affiliation[e]{Section of Mathematics, Université de Genève, Switzerland}

\emailAdd{pygdfm@leeds.ac.uk}
\emailAdd{bruno.mera@tecnico.ulisboa.pt}
\emailAdd{Paulo.Mourao@unige.ch}
\emailAdd{jpnunes@math.tecnico.ulisboa.pt}
\emailAdd{jmourao@math.tecnico.ulisboa.pt}

\makeatletter
\gdef\@fpheader{}
\makeatother

\abstract{We study the change of
the Laughlin
states
under large deformations of 
 the geometry of the
sphere and the 
plane, 
associated with 
 Mabuchi geodesics on the space of  metrics
with Hamiltonian $S^1$--symmetry.

For geodesics associated with 
the square of the 
 symmetry generator,
as the geodesic time goes to infinity,
the geometry of the sphere becomes that of a thin cigar collapsing to a line and
the Laughlin states become concentrated on a discrete set 
of $S^1$--orbits, corresponding to Bohr-Sommerfeld orbits of geometric quantization.

The lifting of the Mabuchi geodesics to the bundle of quantum states, to which the Laughlin states belong, is achieved via generalized coherent state transforms, which correspond to the KZ parallel transport of 
Chern-Simons theory.

}

\begin{document}
\maketitle
\section{Introduction}
\label{sec: introduction}
The importance of the surface geometry dependence of the Fractional Quantum 
Hall effect has been emphasized and studied
by several authors   
\cite{haldane:83,hal:85,jan:lie:sei:08, klevtsov:14,ferrari:klevtsov:14,can:laskin:wiegmann:15, laskin:can:wiegmann:15,Johri:Papic:Schmitteckert:Bhatt:Haldane:16, gromov:geraedts:bradlyn:17, 
liu:gromov:papic:18, klevtsov:19, murugan:shock:slayen:19, papicetal:21, liu:balram:papic:gromov:21, klevtsov:zvonkine:21}. In the present paper, we consider a deformation of the background metric of an oriented surface with an Hamiltonian 
$S^1$--symmetry, 
and with a uniform (with respect to the area $2$-form) magnetic field. This endows the configuration space of the system with the structure of an effective phase space whose K\"{a}hler quantization determines the lowest Landau level (LLL). Natural families of deformations of the physical metric are provided by paths which are geodesics relative to the Mabuchi metric on the space of K\"{a}hler structures~\cite{semmes:1992,donaldson:1999}. The latter are generated by imaginary time Hamiltonian flows~\cite{Mou:Nun:2015}.

We use methods from geometric quantization to study the evolution, 
along the geodesic family
of
deformations, of the single-particle and many-particle states, in particular for the case of Laughlin states. This is done via a generalized coherent state 
transform (GCST)~\cite{kirwin:mourao:nunes:13b}, which lifts the Mabuchi geodesics to the bundle of quantum states.
In the flat case with quadratic Hamiltonians,
these transforms are projectively unitary and coincide with the parallel transport with respect to the 
Knizhnik–Zamolodchikov-Hitchin (KZH) connection \cite{hitchin:90, axelrod:dellapietra:witten:91}. In the case
of cotangent bundles of Lie groups of compact type, these transforms are also,
for appropriate choices of Hamiltonians,
unitary transforms which correspond to 
the classical Segal-Bargmann transform 
and to the Hall coherent state transform \cite{hall:94, florentino:matias:mourao:nunes:05}.

The Mabuchi geodesics on the  space of
K\"ahler metrics on $\Sigma$ which we consider correspond to 
Hamiltonian motion in imaginary time generated by $H = x^2/2$, where $x$ is the function generating the Hamiltonian $S^1$--symmetry.
The geodesics exist for infinite
geodesic (= imaginary-Hamiltonian) time $s$. As $s \to +\infty$,
the lengths of the $S^1$--orbits
converge to zero and the scalar curvature converges, in the case of the plane and the sphere, to a $\delta$--function supported on the $S^1$--fixed points. In both cases, as the geometry becomes more and more deformed, the holomorphic fractional Laughlin wave functions converge to distributional wave functions supported on the ``Bohr-Sommerfeld leaves'', corresponding to integer values of $x$. 

This asymptotic behaviour has important consequences for the density profiles for extreme deformations of the geometry. In  particular, this geometric quantization inspired evolution of single-particle and many-particle LLL states leads to characteristic asymptotic features in the density profiles in the case of fractional filling factor. The study of quantum Hall states in different geometries has had a long history, 
 in particular in~\cite{haldane:83, hal:85}, Haldane considered the case of the sphere with the round metric and the torus with a flat metric, respectively; in~\cite{jan:lie:sei:08} the cylinder with standard metric was also considered. Recently, the change of the quantum Hall states with respect to small deformations of the geometry has received a lot of attention because it was deemed to be relevant for a full diagnosis of topological order~\cite{hu:liu:she:hal:21}. In particular, in~\cite{Johri:Papic:Schmitteckert:Bhatt:Haldane:16} the effect of these small deformations of the geometry of the cylinder was considered. Klevtsov~\cite{ferrari:klevtsov:14,klevtsov:14,klev:16,klevtsov:19} has given a general prescription for constructing Laughlin states in general Riemann surfaces which takes into account holomorphic data on the surface. 
 When changing the geometry, this pescription amounts to taking into account the
 change in the holomorphic structure
 of the wave functions.
  In the language of geometric quantization, this is implemented by evolution in imaginary time under the \emph{prequantum operator}. The latter is only part of the full GCST, and, although it  provides a transport along a path of geometries, it fails, among other things, to correct the norms of the 1-particle states as given
by the evolution with respect to the quantum operator. This leads to different asymptotic density profiles in the two approaches -- an interesting effect that could be probed experimentally.

In this work, we consider the deformation of LLL states for surfaces with a toric structure, that is with a K\"ahler structure invariant under an $S^1$--action.
More concretely, we consider the case of the
sphere and the plane
 with $S^1$-- invariant Riemannian metric, and study the associated particle densities for both the integer quantum Hall effect and the filling fraction $\nu=1/3$ generalized Laughlin states, along Hamiltonian evolution in imaginary time $\tau=-is, s>0$. As $s$ becomes large, the densities  concentrate on the  Bohr-Sommerfeld leaves, and have scalar coefficients that depend on  both the combinatorial properties of the Laughlin states and on the K\"ahler metric geometry of the surface.

The manuscript is organized follows. In Section \ref{sec:preliminaries}, we review aspects of the K\"ahler geometry of toric manifolds, Hamiltonian flows in complex time and geometric quantization. In Section \ref{sec: the lowest Landau level for deformed geometries}, we study the evolution of one-particle states along families of deformed geometries with $S^1-$symmetry. The corresponding Laughlin states are studied in Section \ref{chap:implement} and in Section \ref{densityprofile} we present the density profiles for the deformed geometries. We end with some conclusions in Section \ref{chap:conclusion}.
Our results contain and expand the results
obtained in the MSc Theses \cite{gabriel, mourao}.

\section{Preliminaries}
\label{sec:preliminaries}
In this section, we review basic material from the symplectic and complex geometry of toric K\"ahler surfaces and of their geometric quantization. We also review the application of Hamiltonian flows in imaginary time to the deformation of toric K\"ahler  structures. This will provide the main ingredients to describe the one-particle states for the quantum Hall effect in deformed geometries and later on also the associated many-particle Laughlin states.
\subsection{Toric K\"{a}hler surfaces}
\label{subsec: Kahler toric}
Recall that a toric K\"{a}hler surface is a connected $2-$dimensional 
K\"ahler manifold $(M,\omega, J, \gamma)$, with a compatible triple consisting 
of a symplectic structure $\omega$, complex structure $J$ and Riemannian metric $\gamma$, 
such that the three structures are invariant under an Hamiltonian $S^1-$action on $M$.

The image of the resulting moment map $\mu\colon M\to P$ is always a Delzant polytope (see, for instance, \cite{guillemin:1994}). In the two-dimensional case, that we are considering, 
the only three possibilities for  $P$ are the following:
 \begin{itemize}
     
     \item[\it (i)]
 Closed 
 interval $P=[a,b] \subset \mathbb R \, , (a<b)$, if $M = \mathbb{CP}^1 \cong S^2$.

 \item[\it (ii)]
 Half--line: $P={\mathbb R}_{\geq a} = [a, \infty) $, if
 $M = {\mathbb C}$.

\item[\it (iii)]
 $P=\mathbb R$, if 
 $M = {\mathbb C}^* = {\mathbb C} \setminus \{0\}$.     
    
 \end{itemize}

There are two natural systems of coordinates, that we recall below, and that allow one to encode 
the toric K\"ahler structure of the surface by means of a convex function on the polytope.
First, one has the toric holomorphic coordinate. This is given by $w=e^z$, $z=y+i\theta$, on the dense subset $M^0=\mu^{-1}\left(\text{int}(P)\right)
\cong \mathbb{C}^*$ in which the $S^1-$action takes the form \cite{guillemin:1994,abreu:03}
\begin{equation}
e^{it}\cdot\left(e^{y+i\theta}\right)=e^{y+i(\theta+t)},\quad e^{it}\in S^1, \, y\in \text{int}(P),\ 0\leq \theta< 2\pi.
\end{equation}
Furthermore, since $\omega$ is $S^1-$invariant, one can choose the K\"{a}hler potential $\kappa$ over $M^0$ to also be $S^1-$invariant, i.e. $\kappa=\kappa(y)$, implying that
\begin{equation}
    \omega\big|_{M^0}=i\partial\overline{\partial}\kappa=i\frac{\partial^2\kappa}{\partial z\partial \overline{z}}\,dz\wedge d\overline{z}=\frac{i}{4}\kappa''\,dz\wedge d\overline{z}=\frac{1}{2}\kappa''\,dy\wedge d\theta
\end{equation}
and
\begin{equation}
    \gamma\big|_{M^0}=\omega(\cdot,J\cdot)\big|_{M^0}=\kappa''dy^2,
\end{equation}
or, in matrix form,
\begin{equation}
    \gamma=\begin{bmatrix}
    \kappa''& 0\\
    0 & \kappa''
    \end{bmatrix}.
\end{equation}
 All possible K\"{a}hler toric structures with fixed complex structure are captured by the definition of this strictly convex function $\kappa$.

We now go back to the above coordinates $(y,\theta)\colon M^0\cong\mathbb{R}\times S^1$. Since $\kappa$ is strictly convex, one can define, via Legendre transformation,
\begin{equation}
    x=\frac{\partial \kappa}{\partial y}.
\end{equation}
Obviously, we have
\begin{equation}
    e^{it}\cdot (x,\theta)=(x,\theta+t),\quad e^{it}\in S^1
\end{equation}
and, in fact, it can be shown that, in coordinates $(x,\theta)$, known as action-angle coordinates, $\omega$ takes the simple form 
\begin{equation}\label{AAomega}
    \omega|_{M^0}=dx\wedge d\theta.
\end{equation}
From this, it becomes clear that $x$ is actually a moment map for this action and thus we get a diffeomorphism
\begin{equation}
    M^0\cong\text{int}(P)\times\mathbb{T}^n.
\end{equation}
Furthermore, we can define $g:\text{int}(P)\to\mathbb{R}$ as the Legendre transform of $\kappa$, which is known as the \emph{symplectic potential}. Then
\begin{equation}
    y=\frac{\partial g}{\partial x}\implies \frac{\partial}{\partial x}=g''\frac{\partial}{\partial y}
\end{equation}
and, applying this change of basis to $J$, we get
\begin{equation}
    J=\begin{bmatrix}
    \left(g''\right)^{-1} & 0\\
    0 & 1
    \end{bmatrix}
    \begin{bmatrix}
    0 & 1\\
    -1 & 0
    \end{bmatrix}
    \begin{bmatrix}
    g'' & 0\\
    0 & 1
    \end{bmatrix}=
    \begin{bmatrix}
    0 & \left(g''\right)^{-1}\\
    -g'' & 0
    \end{bmatrix}.
\end{equation}
Consequently, the K\"ahler metric becomes, in matrix form,
\begin{equation}\label{AAmetric}
    \gamma=\begin{bmatrix}
    g'' & 0\\
    0 & \left(g''\right)^{-1}
    \end{bmatrix}.
\end{equation}\par

Furthermore, from the general theory~\cite{guillemin:1994,abreu:03}, it is known that for every $P$ of the form
$$
P = \left\{x \in {\mathbb R}^n \, :  \, \ell_i(x) = \langle\nu_i , x \rangle + a_i \geq 0 \, , \quad i = 1, \dots, d\right\}  \, ,
$$
where the $\nu_i$'s are the exterior pointing normals to the facets and $d$ some non-negative integer, there is a ``canonical'' symplectic potential $g_P$ with fixed singular behaviour on its boundaries and given by
\begin{align}
\label{eq21}
g_P(x) = \frac{1}{2}\sum_{i=1}^d \ell_i(x)\log \ell_i(x).
\end{align}
For the case of toric K\"{a}hler surfaces considered in this work, we have $n=1$ and we will use this ``canonical'' symplectic potential as defining the unperturbed geometry for the three different possible choices of $P$.

From \cite{abreu:03}, it follows that one can deform the geometry by adding to a symplectic potential any function $H$, smooth on the whole of $P$
$$
g \mapsto g + H, 
$$
such that the sum remains convex. These deformations preserve the symplectic form since $\omega$ does not depend on $g$ (cf. Eq.~\eqref{AAomega}). These are the type of deformations we will find in the following discussion by considering Hamiltonian flows in imaginary time.

\subsection{Hamiltonian flows in imaginary time}
\label{subsec: Hamiltonian flows in imaginary time}
First, let us recall a few facts from the theory of flows in complex time. Let $(M,\omega,J,\gamma)$ be a compact K\"{a}hler manifold, where $\omega$ is the symplectic structure, $J$ is the complex structure and $\gamma$ is the Riemannian metric and such that all the three structures are real analytic. Let $H$ be a real analytic Hamiltonian function on $M$ and denote by $X_{H}$ the associated Hamiltonian vector field, i.e., $\iota_{X_{H}}\omega=dH$. It follows from the theory of ODEs that the flow of $X_H$, $\varphi^{X_H}_t$, is also real analytic. Additionally, we have that, within regions of convergence \cite{MR3455877},
\begin{equation*}
    f\in C^\omega(M)\implies \left(\varphi_t^{X_H}\right)^*f=e^{tX_H}f=\sum_{k=0}^\infty\frac{X_H^k(f)}{k!}t^k,
\end{equation*}
which, for some $T>0$ and $\tau\in\mathbb{C}$, $|\tau|<T$, we can analytically continue to get
\begin{equation}\label{cflow}
 e^{\tau X_H}f=\sum_{k=0}^\infty\frac{X_H^k(f)}{k!}\tau^k \in C^\omega(M)\otimes\mathbb{C}.
\end{equation}

We can, in particular, apply (\ref{cflow}) to local $J$-holomorphic coordinates on $M$, such that for small enough $\tau$, one obtains 
new local functions
\begin{equation}\label{ccoord}
     z^j_{\tau}=e^{\tau X_{H}}z^j.
\end{equation}
It is shown in \cite{Mou:Nun:2015} that this local construction in fact  globalizes to produce a well-defined new global complex structure $J_\tau$, such that $z_\tau^j$ are local $J_\tau$-holomorphic coordinates, and that 
it defines a diffeomorphism $\varphi^{X_H}_\tau : M\to M$, such that $(\varphi^{X_H}_\tau)^*J_\tau = J.$
Moreover, one obtains a new global K\"ahler structure $(M,\omega, J_{\tau},\gamma_{\tau})$ where the original symplectic form is unchanged.
Therefore, complex time Hamiltonian flows give us a systematic way of deforming K\"{a}hler structures.

We now apply the concepts discussed above specifically to the case of $2$-dimensional Kähler toric manifolds. For that, let us consider a Kähler toric surface $(M,\omega,J,P)$ as described in Subsection~\ref{subsec: Kahler toric}, with toric holomorphic coordinate $w=e^{z}$, with $z= y+i\theta$ and action-angle coordinates $(x,\theta)$, and a Hamiltonian $H=H(x)$ on $M$. Then

\begin{equation}
    X_{H}=-H'(x)\frac{\partial}{\partial\theta}
\end{equation}
and applying (\ref{ccoord}) to the toric holomorphic coordinate, one obtains
\begin{equation}\label{eqws}
    w_s= \left(e^{\tau X_H}\cdot w\right)\Big|_{\tau=-is}=\left(e^{y+i(\theta-\tau H'(x))}\right)\Big|_{\tau=-is}=e^{y+sH'(x)+i\theta}
\end{equation}
which also motivates the definitions $y_s= y+sH'(x)$ and $z_s= y_s+i\theta$.

Since we are keeping the symplectic structure fixed, the same action-angle coordinates are applicable and thus
\begin{equation}\label{eqgs}
    y_s=\frac{\partial g_s}{\partial x} \iff y+sH'(x)=\frac{\partial g_s}{\partial x}\iff g_s=g+s\left(H(x)+c\right)
\end{equation}
where $c\in\mathbb{R}$ is an integration constant, which we take to be zero since it does not alter the system.
Consequently, the deformed metric has the  form \cite{abreu:98,MR1969265}
\begin{eqnarray}
  \label{equ23b}
   \gamma_s(x) &=& g_s''(x) dx^2 + \frac 1{g_s''} \, d\theta^2 
\end{eqnarray}
with scalar curvature given by Abreu's formula \cite{MR1969265},
\begin{equation}
\label{equ24b} 
Sc(x) = - \left(\frac 1{g_s''(x)} \right)''
\, .
\end{equation}
As for the K\"{a}hler potential, since it is the Legendre dual of $g$, it is given by $\kappa_s=xy_s-g_s$.

\subsection{Geometric quantization}
\label{sec: geometric quantization}
Given a symplectic manifold $(M,\omega)$, describing the phase space of some classical system, the mathematical problem of finding its quantization can be addressed in a geometric framework known as geometric quantization. Here, one assumes that there exists a complex line bundle $L\to M$, equipped with a compatible connection 
$\nabla$ and Hermitian structure such that the curvature of $\nabla$ is $-\frac{i}{\hbar}\omega$. $L$ is known as the \emph{pre-quantum line bundle}.  Note that this places an integrality condition $\left[\frac{\omega}{2\pi\hbar}\right]\in H^2(M,\mathbb{Z})$. 

Geometric quantization produces a Hilbert space of quantum states, $\mathcal{H}_P$, depending on the choice of a polarization: an integrable Lagrangian distribution $P\subset TM\otimes \mathbb{C}$. The Hilbert space is then
$$
\mathcal{H}_P = \overline{\{s\in\Gamma(M,L): \nabla_{X}s=0, \text{ for  any } X\in\Gamma(M,\overline{P}) \text{ and } ||s||_{L^2}^2<\infty\}},
$$
where the completion is with respect to the $L^2$-norm.
When $(M,\omega)$ is K\"ahler, a natural choice is to take $P$ to be the holomorphic tangent bundle of $M$, $P=T^{(1,0)}M$, so that one obtains $\mathcal{H}_P=H^0(M,L)$.

Given $f\in C^\infty(M)$, its prequantum operator is a complex-linear map $Q_{\text{pre}}(f):\Gamma(M,L)\to \Gamma(M,L)$ defined by 
\begin{equation}\label{preq}
    Q_{\text{pre}}(f)= i\hbar\nabla_{X_f}+f.
\end{equation}

Note that, for general $f$, $Q_{\text{pre}}(f)$ will not preserve the space of quantum states $\mathcal{H}_P$. This will happen if $\left[X_f,\overline{P}\right]\subset\overline{P}$.

Above, we have described how Hamiltonian flows in imaginary time $\tau$ can be used to deform the complex structure of a K\"ahler manifold. As the complex structure changes along the deformation, we will obtain a family of K\"ahler polarizations $P_\tau$ and the corresponding Hilbert spaces of quantum states, $\mathcal{H}_{P_\tau}$ will vary in the space of smooth sections of $L$, $\Gamma(M,L)$. To relate these Hilbert spaces to each other, we need, given a complex time Hamiltonian flow $\varphi^{X_H}_\tau$ as described in Section~\ref{subsec: Hamiltonian flows in imaginary time}, a way to lift this deformation to the system quantized via geometric quantization. Since the initial polarization $P$ changes to a polarization $P_\tau$, we will naturally get a map
\begin{equation*}
    U_\tau\colon \mathcal{H}_{P}\to\mathcal{H}_{P_\tau}.
\end{equation*}
In this work we considered evolutions in imaginary time $\tau=-is$ and the appropriate choice turns out to be a GCST of the form (see \cite{kirwin:mourao:nunes:13b})
\begin{equation}
\label{eq: evolutionop}
    U_s=\left(e^{\frac{i}{\hbar}\tau Q_{\text{pre}}(H)}e^{-\frac{i}{\hbar}\tau Q(H)}\right)\Big|_{\tau=-is},
\end{equation}
where $Q(H)$ is an appropriate quantum operator for $H$ that we will describe explicitly, for the case of toric deformations of the initial toric K\"ahler structure, below.
Unlike $Q_{\text{pre}}(H)$, the quantum operator $Q(H)$ is a choice of quantization of the classical observable $H$ that preserves $\mathcal{H}_{P_0}$. For example, if $h$ is a classical observable such that $X_h$ preserves ${P_0}$ and if $H = h^2$ then a natural choice will be $Q(H) = (Q_{\text{pre}}(h))^2$.
Note that the operator
\begin{equation}
\label{eq: preevolutionop}
    U^{\text{pre}}_s=e^{\frac{i}{\hbar}\tau Q_\text{pre}(H)}
\end{equation}
is a generalization of time evolution in quantum mechanics and we know it does not preserve the Hilbert space of polarized sections, but rather maps $J$-polarized states to $J_\tau$-polarized states (cf. Eq.\eqref{prequantumcalculation}) in a natural way. In fact, the Laughlin states used by Klevtsov (see \cite{klevtsov:19}) can be replicated by simply considering this operator.\par 
However, this evolution by prequantization of the Hamiltonian is highly non-unitary. The inclusion of $e^{-\frac{i}{\hbar}\tau Q(H)}$, which we know to preserve the Hilbert space, ``reverts" the effect of the prequantum evolution operator on quantum states without preventing the geometry deformation of our system ( Eq.~\eqref{quantumcalc}) and restores unitarity asymptotically. We will analyse closely the non-unitarity of~\eqref{eq: preevolutionop}, as well as the consequences of taking the above GCST instead, in Proposition~\ref{limitgcst} and also Section in \ref{densityprofile} where particle density profiles will be presented.

\section{The lowest Landau level for deformed geometries}
\label{sec: the lowest Landau level for deformed geometries}
We consider charged particles living on a K\"ahler surface $(M,\omega,J,\gamma)$ and subject to a uniform external magnetic field. By a uniform magnetic field we mean that the Faraday $2-$form $F$  is proportional to the area form of the surface, i.e., $F= B\omega$.
The curvature of $L$ is $\Omega= -i\frac{qF}{\hbar}=-i\frac{\omega}{\hbar_{\text{eff}}}$, where $q$ is the charge of the carriers and $B$ is the magnetic field, and where we introduced an effective Planck's constant, denoted $\hbar_{\text{eff}}$, associated to the configuration space which, effectively, behaves like a phase space in the physics of the lowest Landau level. Thus, $\hbar_{\text{eff}}=\hbar /(q B):=\ell_B^2$ where $\ell_B$ is the so-called magnetic length. Below, in the case of the sphere, the dimension of the one-particle Hilbert space is $N=h^0(L)=c_1(L)+1$ where $c_1(L) = qBA/h$, in which $A$ denotes the area of the surface. Thus, $N=1+qBA/h.$
The single-particle Hilbert space of the quantum theory is described by the square integrable sections of the electromagnetic line bundle $L\to M$, over which the electromagnetic gauge field describes a unitary connection whose curvature $2-$form is given by $-\frac{i}{\hbar_{\text{eff}}}\omega$. The single-particle Hamiltonian is described by the Bochner Laplacian, whose groundstate subspace, known in the physics literature as the lowest Landau level (LLL), is described by $H^0(M,L)$~\cite{klev:16}. It is now clear the relation to geometric quantization: the lowest Landau level is nothing but the Hilbert space of quantum states for the K\"ahler quantization of $M$.

In this section, we describe the LLLs on the  sphere and the plane with (toric) geometries deformed by Hamiltonian flow in imaginary time. 
We begin with the round sphere $S^2$ in Section \ref{ss211}, for which the process is described in detail. We determine explicit variations of the structure, along with holomorphic states and Hermitian product.\par
In section~\ref{ss212}, we make the analogy with the case of the plane, since we will also compute evolution of density profiles for this system in Section \ref{densityprofile}. For more details on deformations of the plane see \cite{gabriel}.\par
In Section \ref{gcstsectionn}, we explicitly apply GCST to the states described previously, namely in the limit $s\to\infty$. We also compare this evolution with the one obtained by using just the prequantum evolution operator of Eq.~\eqref{eq: preevolutionop}, as this is the evolution operator that produces the Laughlin states used by  Klevtsov (see \cite{klevtsov:19}).

\subsection{One particle states on the deformed  sphere}
\label{ss211}\leavevmode\par
We start with a first analysis of the sphere $M=S^2\cong\C P^1$.\par
As mentioned in Section \ref{subsec: Kahler toric}, the corresponding polytope in this case is $P=[a,b]\subset\R$, where $a<b$. Right away, we can use (\ref{AAomega}) to compute the symplectic area
\begin{equation}
    \int_M\omega=2\pi(b-a).
\end{equation}
One immediately concludes that, for the system to be quantizable, we must have
\begin{eqnarray}
  \frac{b-a}{\hbar_{\text{eff}}}\in\Z.
\end{eqnarray}
For simplicity of computations, we will now assume Planck units so that $\hbar_{\text{eff}}=1$ and hence $P=[a,a+N]$, for some $N\in\N$. The symplectic potential (\ref{eq21}) then reads
\begin{equation}
\label{eqq22aa}
g_P(x)=\frac{1}{2}\left(\left(x-a\right)\log\left(x-a\right)+\left(N+a-x\right)\log\left(N+a-x\right)\right)
\end{equation}

Furthermore, as we have seen in Section~\ref{sec:preliminaries}, the toric holomorphic coordinates are obtained from the action angle ones, by Legendre transform, hence we get
\begin{equation}\label{ws}
    x\mapsto y=\frac{\partial g_P}{\partial x}=\frac{1}{2}\log\left(\frac{x-a}{N+a-x}\right)\implies w=e^{y+i\theta}=\sqrt{\frac{x-a}{N+a-x}}e^{i\theta}
\end{equation}\par
Now, if we consider an imaginary time Hamiltonian flow induced by $H(x)=\frac{1}{2}x^2$, then the results of Eqs.~\eqref{eqws}, \eqref{eqgs} and \eqref{equ23b} from Section~\ref{sec:preliminaries}  yield
\begin{align}
    \label{ysESFERA}y_s&=y+sx\\
    \label{gsESFERA}g_s&=g_P+s\frac{x^2}{2}\\
    \label{wsESFERA}w_s&=\sqrt{\frac{x-a}{N+a}}e^{sx+i\theta}\\
    \label{eq23b}\gamma_s&= \left( \frac 1{2\left(x-a\right) \left(N+a - x \right)} + s\right)  dx^2 +\, \left( \frac 1{2\left(x-a\right) \left(N+a - x \right)} + s\right)^{-1} 
    d \theta^2
\end{align}
And the scalar curvature (\ref{equ24b}) becomes
\begin{equation}
\label{eq24b} 
Sc(x) = - \left(\frac{1}{\frac{1}{2\left(x-a\right)\left(N+a-x\right)}+s}\right)''
\, ,
 \end{equation}
which is constant at $s=0$ and concentrates
around the poles as $s \to \infty$ (see \cite{kirwin:mourao:nunes:13a} for more details).

Writting the induced metric as
\begin{equation}\label{roundsphere}
  \gamma_s(x)=g_s''\,dx^2 + \frac{1}{g_s''} \, d\theta^2=\left(\frac 1{2u \left(N - u \right)} +s\right)\, du^2 + \left(\frac{1}{
    2 u\left(N-u \right)}+s\right)^{-1}
    d \theta^2
\end{equation}
where we made the substitution $u=x-a$, we see that, if $s=0$, it is the Fubini-Study metric and, in general, it does not depend on the choice of $a$. Furthermore, if $x$ is a moment map, then so is $x-a$. This translates the polytope by $a$, but (\ref{roundsphere}) shows that (assuming we are using the canonical symplectic potential) the Kähler structure remains entirely equivalent. We can also see that the choice of $N$ amounts to a simple rescaling of the sphere and the imaginary time unit. It is then common practice to let these (so far irrelevant) parameters denote a specific choice of line bundle $L$ and basis of holomorphic sections in a very natural way. We now present this construction (more details and the generalization for $n$-dimensional Delzant polytopes can be found in \cite{MR1234037}).\par

First, consider a formal sum
\begin{eqnarray}\label{divisor}
  D^L := \lambda_1^LD_1+\lambda_2^LD_2
\end{eqnarray}
where $D_1=x^{-1}(a)$ and $D_2=x^{-1}(a+N)$. By taking $\lambda_1^{L},\lambda_2^{L}\in\mathbb{Z}$, we have what is called a toric divisor in algebraic geometry. It is a notation that is used represent sections of some line bundle $L$ on $M$ that can be given in local holomorphic coordinates around each $D_j$ by $z_s^{\lambda_j^L}$. This proves to be very convenient, as it can be shown that there exists only one such line bundle and only one section $\sigma_{D^L}$ of that line bundle, up to multiplication by a constant, with divisor $D^L$. We will now construct $L$ explicitly.

The toric holomorphic coordinates $w_s$ given in (\ref{ws}) cover $S^2\setminus\{D_1,D_2\}$ and can be trivially extended to $U_1\defeq S^2\setminus\{D_2\}$. These are simply stereographical projection coordinates and thus the atlas is completed by considering a similar map $\tilde{w}_s\colon U_2\to\C$, with $U_2:= S^2\setminus\{D_1\}$ and $w=\tilde{w}^{-1}$ in the intersection $U_1\cap U_2$. If the section $\sigma_{D^L}$ is given by $w_s^{\lambda_1^L}$ in $U_1$ and $\tilde{w}_s^{\lambda_2^L}$ in $U_2$, then
\begin{equation}
    \tilde{w}_s^{\lambda_2^L}=f_{1,2}\,w_s^{\lambda_1^L}\iff f_{1,2}=w_s^{-\left(\lambda_1^L+\lambda_2^L\right)}
\end{equation}
in $U_1\cap U_2$, where $f_{1,2}$ is the corresponding transition function for $L$, which uniquely defines the line bundle up to equivalence. Evidently, the space of holomorphic sections of this line bundle is given by $H^0(M,L):=\text{span}_{\mathbb{C}}\{w_s^{m}\sigma_{D^L}: \text{div}\left(w_s^m\sigma_{D^L}\right)\geq 0\}$, with
\begin{equation}\label{eq313}
\text{div}\left(w_s^m\sigma_{D^L}\right)\geq 0\iff m+\lambda_1^{L}\geq 0 \text{ and } -m+\lambda_2^L\geq 0.
\end{equation}
where $\text{div}$ denotes the divisor of a given section. Thus we see that there is a bijection between this $S^1$-equivariant basis of $H^0(\mathbb{CP}^1,L)$ and the integral points (we want $m\in\Z$ for the sections to be single-valued) of the Delzant polytope $P_L=\left[-\lambda_1^L,\lambda_2^L\right]$.\par
When considering half form however, it becomes appropriate to have $\lambda_j^L\in\Z+\frac{1}{2}$. This is because, since $w_s=e^{z_s}$, we have that $\text{div}(\sqrt{dz_s})=\text{div}\left(\sqrt{\frac{dw_s}{w_s}}\right)=-\frac{1}{2}(D_1^L+D_2^L)$, meaning that $\sigma_{D^L}\otimes \sqrt{dz_s}$ is single-valued only when the $\lambda_j^L$'s are half-integers. And although the rest of the argument is still valid, (\ref{eq313}) now becomes
\begin{equation}
\text{div}\left(w_s^m\sigma_{D^L}\otimes\sqrt{dz_s}\right)\geq 0\iff m+\lambda_1^{L}-\frac{1}{2}\geq 0 \text{ and } -m+\lambda_2^L-\frac{1}{2}\geq 0.
\end{equation}
and so once again we obtain a bijection with integral points of the Delzant polytope $P_L=\left[-\lambda_1^L,\lambda_2^L\right]$.\par
Hence, in either case, the divisor (\ref{divisor}) defines a line bundle $L$ on $M$ and a basis of $H^0\left(\mathbb{CP}^1,L\right)$ indexed by the integer points of $P_L$. Naturally, we let $a=\lambda_1^L$ and $N=\lambda_2^L+\lambda_1^L$, so that, as promised, the choice of polytope determines the line bundle we are considering. Note that, in this way, the choice of area of our surface (i.e. choice of $N$) determines how many one-particle states our system has.\par
In our case, we will be considering half-form correction and taking $a=-\frac{1}{2}$ for simplicity. Then $\lambda_2^L=N-\frac{1}{2}$ and, denoting $\sigma_{D^L}$ by $\sigma_{P,s}$, we write the elements of the basis as
\begin{align}
\label{eqq23}
\sigma^m_s=w^m_s\sigma_{P,s}\otimes \sqrt{dz_s},\ m\in P_L\cap\mathbb{Z}=\{0,...,N-1\},
\end{align}

Next we introduce an Hermitian structure on $L$ without half-form correction. These results are still applicable when including the half-form correction, by considering the product with the norm of the half-form correction (cf. Eq.~\eqref{halfformherm}) (see \cite{kirwin:mourao:nunes:13a}). For that, let $\nu\colon\Gamma(L)\times\Gamma(L)\to C^\infty(M,\mathbb{C})$ denote the pointwise Hermitian product with $\norm{\cdot}$ the corresponding norm. Note that, since $L$ is a line bundle, it is enough to give the value of $\norm{\sigma_{P,s}}$, since the Hermitian product then follows by skew-linearity. If $\tilde{s}_0$ is a local unitary trivializing section, then the connection is given locally by
\begin{equation}\label{cov}
    \nabla_{X}\left(f\tilde{s}_0\right)=\left(X(f)-i\theta(X)f\right)\tilde{s}_0.
\end{equation}
Since $\omega=i\partial\overline{\partial}\kappa_s$, we can take
\begin{equation}
    \theta=i\left(\overline{\partial}-\partial\right)\kappa_s.
\end{equation}
Therefore, the holomorphic sections we determined before to be given by $H^0(M,L)=\text{span}_{\mathbb{C}}\left\{w^m_s\sigma_{P,s}\colon m=0,\ldots,N-1\right\}$ are, locally, solutions of
$$
\nabla_{\frac{\partial}{\partial\overline{z}}}fs_0=0,\, \,\,\,{\rm or} \,\,\,\,\,\frac{\partial f}{\partial\overline{z}}+\frac{\partial \kappa_s}{\partial\overline{z}}f=0.
$$
Defining $F=e^{\kappa_s}f$, the above reduces to
\begin{equation}
    \frac{\partial F}{\partial\overline{z}}=0.
\end{equation}
Hence
\begin{equation}\label{sigmaps}
    \sigma_{P,s}=e^{-\kappa_s}s_0\implies\nabla_{\frac{\partial}{\partial\overline{z}}}w^m\sigma_{P,s}=0
\end{equation}
and so we can define the hermitian structure by $\norm{\sigma_{P,s}}=e^{-\kappa_s}$. With half-form correction, this becomes
\begin{equation}\label{halfformherm}
    \norm{\sigma_{P,s}\otimes \sqrt{dz_s}}^2= \exp(-2\kappa_s)\norm{dz_s}
\end{equation}

Summarizing, the $S^1$-equivariant 
basis of the space of holomorphic
sections is given by
\begin{equation}
\label{eqq29}
\left\{  
\begin{array}{rcl}
\sigma_s^m &=& \widetilde{\sigma}^m_s\otimes \sqrt{dz_s},\qquad  m\in P\cap\mathbb{Z}=P_L\cap\mathbb{Z}=\{0,...,N-1\}\\
   \widetilde{\sigma}^m_s &=& w_s^m \, \sigma_{P,s} \, ,\\
   w_s &=& e^{z_s} 
   = e^{y_s + i \theta} \, ,\\
   y_s&=&g_s' \,,   \end{array}
\right.  
\end{equation}
where
\begin{equation}
\label{eqq210}
\left\{  
\begin{array}{rcl}
   g_s(x) &=& g_P(x) + \frac s2 \, x^2  \\
   \norm{\sigma_{P,s}\otimes \sqrt{dz_s}}^2 &=& \exp(-2\kappa_s)\norm{dz_s}
   \\
   \kappa_s &=& xy_s -g_s=\kappa_0 +s\left(xH' - H\right) \, , 
 \end{array}
 \right.
   \end{equation}
with $g_P$ defined in (\ref{eqq22aa}).

\subsection{One-particle states on the deformed plane}
\label{ss212}\leavevmode\par

The results of the previous subsection are converted
to the case of the plane 
if we replace
in the formulas of the previous section,
 the polytope $P=[-\frac 12, N-\frac 12]$
by $P=[- \frac 12, \infty )$. 
We see that the divisor $D_2$ does not exist and so we get a bijection between basis for the space of holomorphic sections and $P \cap \mathbb{Z} = \mathbb{Z}_{\geq 0}$. Therefore, as expected, since the 
symplectic area is infinite, 
the space of holomorphic sections is 
infinite dimensional. 
As usual in the quantum Hall effect 
we will describe approximately finite
discs by taking $N$--dimensional subspaces
of the space of holomorphic sections, namely
those corresponding to integers $m$ smaller than $N$. This also has the advantage of making the treatment done in the next chapter for the sphere entirely analogous to the one that would be done for the plane.

From (\ref{eq21}), the canonical symplectic potential 
reads,
\begin{equation}
\label{eqq211}
g_P(x) = \frac 12 \, \left(x+ \frac 12 \right) \log \left(x+ \frac 12\right)  \, ,  
\end{equation}
which corresponds to the standard, $S^1$-equivariant, flat metric on the plane. To simplify
the relation between the symplectic
coordinates $(x, \theta)$ and the
holomorphic ones $w$ let us change 
the canonical potential in (\ref{eqq211}) by a linear term
that leaves the metric unchanged,
\begin{equation}
\label{eqq211b}
g_P(x) \leadsto \tilde g_P(x) = \frac 12 \, \left(x+ \frac 12 \right) \log 2 \left(x+ \frac 12\right) - \frac x2 \, .   \,  
\end{equation}
We consider deformations 
of the geometry induced by 
the same convex function
of the $S^1$--Hamiltonian,
$H(x) = \frac 12 x^2$. 
The expressions (\ref{eqq29}) and (\ref{eqq210}) remain valid 
with $g_P$ replaced by $\tilde g_P$
in (\ref{eqq211}) so that, in particular,
$$
w = e^{\tilde g'_P(x) + i \theta}
= \sqrt{2\left(x + \frac 12\right)} \, e^{i \theta}
$$
with imaginary time evolution given by
$$
w_s = e^{is X_H}(w) = 
\sqrt{2\left(x + \frac 12\right)}\, e^{sx} \, e^{i \theta} = e^{sx} \, w \, , 
$$
and
$$
g_s(x) = \tilde g_P(x) + \frac s2 x^2 \, .
$$
For the metric and its 
scalar curvature we obtain from (\ref{equ23b}) and (\ref{equ24b})
\begin{equation}\label{eq214}
   \gamma_s(x)=\frac{2s\left(x+\frac{1}{2}\right)+1}{2\left(x+\frac{1}{2}\right)}\,dx^2+\frac{2\left(x+\frac{1}{2}\right)}{2s\left(x+\frac{1}{2}\right)+1}
    \, d \theta^2
\end{equation}
and
\begin{equation}\label{eq215} 
Sc(x) =\frac{8s}{\left(2s\left(x+\frac{1}{2}\right)+1\right)^3}
 \end{equation}

\subsection{GCST in the limit {$s \to \infty $}}\label{gcstsectionn}
We now analyse the evolutions described in the previous sections in the limit $s\to\infty$.
Geometric quantization in this scenario has a particular combinatorial flavour, and extreme deformations of the geometry leading to degenerations of the K\"{a}hler structure lead to the convergence of sections, in a distributional sense, to distributional sections localized in integer points of the polytope~\cite{nunes:14, kirwin:mourao:nunes:13a}, as is shown in the next proposition. It also states how both the GCST operator $U_s$ (cf. Eq.~\eqref{eq: evolutionop}) and the prequantum evolution operator $U^{\text{pre}}_s$ (cf. Eq.~\eqref{eq: preevolutionop}) act on one particle states, justifying the claims made in Section~\ref{sec: introduction}. In fact, we see that norms converge as $s\to\infty$ only for the GCST operator.

\begin{proposition}\label{limitgcst} (See \cite{bfmn,kirwin:mourao:nunes:13a})
Consider the GCST operator $U_s$ defined in (\ref{eq: evolutionop}), with $H=\frac{1}{2}x^2$ in action-angle coordinates. Then
\begin{equation}\label{prequantumcalculation}
    U^\text{pre}_s\left(\sigma^m_0\right)=\sigma^m_s
\end{equation}
where $U^\text{pre}_s$ is given by (\ref{eq: preevolutionop}) $\sigma^m_\tau$ denotes the one particle states defined in (\ref{eqq29}). Furthermore,
\begin{equation}\label{quantumcalc}
    U_s\left(\sigma^m_0\right)=e^{-s\frac{m^2}{2}}\sigma^m_s
\end{equation}
and
\begin{equation}\label{finallimit}
    \lim_{s\to\infty}U_s\left(\sigma^m_0\right)=\sqrt{2\pi}\,e^g\delta_m s_0\otimes\sqrt{dx},
\end{equation}
where $s_0$ is a unitary section (cf. Eq.~\eqref{sigmaps}), and $g=g_P$ for the sphere and $g=\tilde{g}_P$ for the plane.
\end{proposition}
\begin{proof}
We have
\begin{equation}\label{xdtheta}
    X_H=-x\frac{\partial }{\partial\theta}.
\end{equation}
Substituting in (\ref{preq}), we get, with respect to the trivializing section defined by (\ref{sigmaps}), 
\begin{align*}
    Q_\text{pre}(H)&=-ix\frac{\partial }{\partial\theta}+i\left(\partial-\overline{\partial}\right)\kappa\left(x\frac{\partial}{\partial\theta}\right)+H\\
    &=-ix\frac{\partial}{\partial \theta}+xd\theta\left(-x\frac{\partial}{\partial \theta}\right)+\frac{x^2}{2}\\
    &=-ix\frac{\partial }{\partial\theta}-\frac{1}{2}x^2
\end{align*}

Furthermore, as (\ref{xdtheta}) does not preserve the complex polarization, we define
\begin{align*}
    Q(H)&=\frac{1}{2} \left(Q_\text{pre}\left(x\right)\right)^2\\
    &=\frac{1}{2}\left(-i\frac{\partial}{\partial\theta}\right)^2\\
    &=-\frac{1}{2}\frac{\partial^2}{\partial\theta^2}
\end{align*}
and so the GCST, given in (\ref{eq: evolutionop}), becomes
\begin{equation}\label{gcstalmost}
    U_s=\exp\left(-isx\frac{\partial}{\partial\theta}-\frac{s}{2}x^2\right)\exp\left(\frac{s}{2}\frac{\partial^2}{\partial\theta^2}\right).
\end{equation}
Evidently, all of the operators in the exponents commute between themselves and thus
\begin{align*}
    U^{\text{pre}}_s\left(\sigma^m_0\right)&=e^{-\frac{s}{2}x^2}e^{-isX_H}e^{m(y+i\theta)}e^{-xy+g}\tilde{s}_0\otimes\sqrt{dz}\\
    &=e^{-\frac{s}{2}x^2}e^{m(y_s+i\theta)}e^{-xy+g}\tilde{s}_0\otimes\sqrt{dz_s}\\
    &=e^{m(y_s+i\theta)}e^{-xy_s+g_s}\tilde{s}_0\otimes\sqrt{dz_s}\\
    &=\sigma^m_s
\end{align*}\par
Furthermore, only the factor $e^{mi\theta}$ is not $S^1$-equivariant, and so Eq.~\eqref{quantumcalc} follows immediately.\par
Finally, using $\norm{\sqrt{dz_s}}=\sqrt{g_s''}$ (see \cite{kirwin:mourao:nunes:13a}) and $\lim_{s\to\infty}\frac{g_s''}{s}=1$ (which is true for both the plane and the sphere), then
\begin{equation}
    \lim_{s\to\infty}\norm{U_s\left(\sigma^m_0\right)}=\lim_{s\to\infty}\,\sqrt{s}e^{-\frac{s}{2}\left(x-m\right)^2}e^{-y(x-m)+g}=\sqrt{2\pi}\,e^{g(m)}\delta_m
\end{equation}
Thus, Eq.~\eqref{finallimit} follows from $\lim_{s\to\infty}\frac{z_s}{s}=x$.
\end{proof}

\section{Evolution of Laughlin states on the sphere}
\label{chap:implement}
We will now study many particle lowest energy states of the quantum Hall effect and their geometry dependence. For simplicity of the exposition we only consider the case of the sphere in the first $2$ sections, as the case of the plane 
is analogous
and we already described the analogies with the sphere in section \ref{ss211}.\par
We first study, in section \ref{iqhe}, the integer quantum Hall effect, in which the ground states are fully filled. Some of the mathematical notions used to represent many-particle states are introduced as well as how they evolve by imaginary time Hamiltonian flow, in particular in the limit $s\to\infty$.\par
In section \ref{fractional}, we describe the FQHE, in which we consider states filling only $1/m$ of ground states, with $m$ odd, since, for these cases, there is a very good approximation for the lowest energy states known as Laughlin states. And using a result by Gerald V. Dunne (\cite{dunne:93}), that gives a Slater decomposition of Laughlin states, we are able to describe their evolution under imaginary time Hamiltonian flow for an arbitrary number of particles.\par
These results are then used to study and discuss density profiles in section \ref{densityprofile}. One very important aspect of this discussion is the comparison with evolutions using only the prequantum evolution operator (\ref{eq: preevolutionop}), which was shown in section \ref{gcstalmost} to correspond to the Laughlin states used by Semyon Klevtsov (see \cite{klevtsov:19}). In fact, the results of the present section showcase a convergence of Laughlin states to specific Slater determinants under the evolution of the prequantum operator, which is entirely a consequence of the non-unitarity of this operator.

\subsection{LLL states of the IQHE}\label{iqhe}\leavevmode\par
The fully filled LLL state corresponds to the unique state in $\Lambda^{N}\mathcal{H}_0$, where $\mathcal{H}_0=H^0(\mathbb{CP}^1,L_0)$, with $L_0$ the holomorphic line bundle described in section \ref{ss211}, corresponding to the $s=0$ initial geometry. It is, intrinsically, an $N$-particle fermionic state.\par
Let $(\mathbb{CP}^1)^N$ denote the $N$-fold Cartesian product of $\mathbb{CP}^1$ and $\pi_i:(\mathbb{CP}^1)^N\to \mathbb{CP}^1$ the $i$th canonical projection. We have a line bundle $L^{\boxtimes N}:=\bigotimes_{i=1}^{N}\pi_i^{*}L_0\to (\mathbb{CP}^1)^N$. Given $N$ sections of $L\to \mathbb{CP}^1$, $s_1,...,s_N$, we have a natural section $s_1\boxtimes s_2\boxtimes...\boxtimes s_N:=\pi_1^{*}s_1\otimes ....\otimes \pi_N^*s_N\in \Gamma(L^{\boxtimes N})$, given by
\begin{align*}
(w_1,...,w_N)\mapsto s_1(w_1)\otimes s_2(w_2)\otimes ...\otimes s_N(w_N).
\end{align*}
Any other section can be written as linear combination of such sections and we learn that, as a vector space, we have $\Gamma(L^{\boxtimes N})\cong \Gamma(L)^{\otimes N}$. Actually, we can identify the section $\pi_1^{*}s_1\otimes ....\otimes \pi_N^*s_N$ as the tensor product $s_1\otimes ...\otimes s_N\in \Gamma(L)^{\otimes N}$. Since we are looking at fermionic particles, we wish to consider only those sections which are completely antisymmetric under permutations of the particles (i.e. permutations of variables $w_j$). Therefore the state belongs to $\Lambda^N\mathcal{H}\subset \Gamma(L)^{\otimes N}$ and corresponds to the wedge product
\begin{align*}
\Psi_{\text{IQHE}}\defeq\sigma_0^0\wedge ...\wedge \sigma_0^{N-1},
\end{align*} 
where $\{\sigma^j_0\}_{j=0}^{N-1}$ is a basis for $\mathcal{H}_0$. As a section of $L^{\boxtimes N}\to \left(\mathbb{CP}^1\right)^N$ it is written as
\begin{align*}
(w_1,...,w_N)\mapsto \sum_{\tau\in S_N}\text{sgn}(\tau)\sigma_0^{\tau(1)}(w_{1})\otimes \sigma_0^{\tau(2)}(w_{2})\otimes ...\otimes \sigma_0^{\tau(N)}(w_{N})
\end{align*}
Given a local trivialization of $L\to \mathbb{CP}^1$ over an open set $U$ provided by a local section $\Tilde{s}_0$, we have an induced local trivialization of $L^{\boxtimes N}\to \left(\mathbb{CP}^1\right)^N$ over $U\times ...\times U$ given by $\Tilde{s}_0^{\boxtimes N}:=\pi_1^{*}\Tilde{s}_0\otimes ... \otimes \pi_N^*\Tilde{s}_0$. If over $U$ we have $\sigma_0^j=f^j\Tilde{s}_0$, for some local functions $f^j$, $i=1,...,N$, then obviously,
\begin{align*}
&\sum_{\tau\in S_N}\text{sgn}(\tau)\sigma_0^{\tau(0)}(w_{1})\otimes \sigma_0^{\tau(2)}(w_{2})\otimes ...\otimes \sigma_0^{\tau(N-1)}(w_{N})\\
&=\Tilde{s}_0^{\boxtimes N}(w_1,...,w_N)\cdot\left(\sum_{\tau\in S_N}\text{sgn}(\tau)f^{\tau(0)}(w_{1})f^{\tau(1)}(w_{2})... f^{\tau(N-1)}(w_{N})\right)\\
&=\Tilde{s}_0^{\boxtimes N}(w_1,...,w_N)\cdot \det \left[\begin{array}{c c c}
f^0(w_1) & \cdots & f^0(w_N)\\
\vdots &  \ddots & \vdots \\
f^{N-1}(w_1) & \cdots & f^{N-1}(w_N)
\end{array}\right],
\end{align*} 
hence it locally looks like the Slater determinant of the associated local representatives. We will denote by $w_{s,j}=\pi_j^{*}w_{s}$ the imaginary time $\tau=-is$ evolved local holomorphic coordinate associated to the $j$-th particle, $j=1,...,N$. We will adopt a similar notation for the other coordinates $y_{s,j}$, $x_{j}$ and $\theta_j$ defined for the single particle case. Analogously to what is done in (\ref{sigmaps}), we consider the local trivialization $1^{\text{U}(1)}\otimes \sqrt{dz_0}$ (where $\sqrt{dz_0}$ denotes the product $\sqrt{dz_1}\otimes\ldots\otimes\sqrt{dz_{N-1}}$ at $s=0$) of $L$ over $U$ so that its local representative is
\begin{align*}
&\prod_{i=1}^{N}e^{-\kappa_0(y_{0,j})}\det\left[\begin{array}{c c c }
1 & \cdots & 1\\
w_{0,1} & \cdots & w_{0,N}\\
\vdots &  \ddots & \vdots \\
w_{0,1}^{N-1} & \cdots & w_{0,N}^{N-1}
\end{array}\right]= e^{-\sum_{i=1}^{N}\kappa_0(y_{0,j})} \prod_{1\leq i<j\leq N}(w_{0,j}-w_{0,i}).
\end{align*}\par
The GCST operator $U_s$	 acts on $\Gamma(L^{\otimes})\cong \Gamma(L)^{\otimes N}$ in a natural way, namely, given $s_1,...,s_N$, we have, in terms of sections of $L^{\boxtimes}$
\begin{align*}
U_s(s_1\boxtimes ...\boxtimes s_N)&=U_s\left(\pi_1^{*}s_1\otimes ...\otimes \pi_N^{*}s_N\right):=\pi_1^{*}(U_s s_1)\otimes ...\otimes \pi_N^{*}(U_s s_N)\\
&=(U_s s_1)\boxtimes ...\boxtimes (U_s s_N),
\end{align*}
or in terms of $\Gamma(L)^{\otimes N}$,
\begin{align*}
U_s(s_1\otimes ...\otimes s_N)=(U_s s_1)\otimes ...\otimes (U_s s_N).
\end{align*}
Consequently, we also get
\begin{equation}
U_s(\sigma_0^0\wedge ... \wedge \sigma_0^{N-1})= (U_s\sigma_0^0)\wedge ...\wedge (U_s\sigma_0^{N-1}).
\end{equation}\par
If $\sigma^m_s$ denotes the states described in section \ref{ss211} (cf. Eq.~\eqref{eqq29}), then it follows from Proposition~\ref{limitgcst}, that
\begin{align}\label{uspsi}
U_s(\Psi_{\text{IQHE}})=e^{-\frac{s}{2}\sum_{m=0}^{N-1}m^2} \sigma^0_s\wedge...\wedge \sigma^{N-1}_s,
\end{align}
and
\begin{align*}
&\lim_{s\to\infty}U_s(\Psi_{\text{IQHE}})\\
&=(2\pi)^{\frac{N}{2}}e^{\sum_{m=0}^{N-1}g(m)}\sum_{\tau\in S_N}\text{sgn}(\tau)(\pi_1^{*}\delta_{\tau(0)}\otimes \sqrt{dx_1})\wedge...\wedge (\pi_N^{*}\delta_{\tau(N-1)}\otimes \sqrt{dx_N}).
\end{align*}\par

\subsection{FQHE Laughlin states}\label{fractional}\leavevmode\par
We will consider now a Laughlin state with $N_e$ electrons and filling fraction $\nu\defeq N_e/N=1/m$ with $m$ odd. We have that $N_e=N/m$, and we will assume that $m$ divides $N$, where $N=\deg L$, as before. At the $s=0$ geometry, i.e. the round sphere geometry (cf. Eq.~\eqref{roundsphere}), the Laughlin state is given in a trivializing open set $U^{N_e}$ by
\begin{align*}
 \Psi_{\text{Laughlin}}\defeq e^{-\sum_{i=1}^{N_e}\kappa_0(y_{0,j})} \prod_{1\leq i<j\leq N_e}(w_{0,j}-w_{0,j})^{m}1^{\text{U}(1)}\otimes \sqrt{dz_0}.	
\end{align*} 
Since $m$ is odd it is clearly skew-symmetric. The pre-factor is just the pointwise norm of the holomorphic section $\sigma_{P,0}^{\boxtimes{N_e}}$ while the remainder is holomorphic. Therefore we conclude that this state belongs to $\Lambda^{N_e}H^0(\mathbb{CP}^1,L_0)$.

We will now consider the GCST operator in the $N_e$-particle sector. From Proposition~\ref{limitgcst}, it is clear that at imaginary time $\tau=-is$, the operator assumes the simple representative 
\begin{align*}
U_{s}=\exp\left(-s\sum_{i=1}^{N_e}\frac{x_i^2}{2}-is\sum_{i=1}^{N_e}x_i\frac{\partial}{\partial \theta_i}\right)\exp\left(\frac{s}{2}\sum_{i=1}^{N_e}\frac{\partial^2}{\partial \theta_i^2}\right),
\end{align*}
where, evidently, the exponentials commute with each other. Hence, a computation entirely analogous to that of Proposition~\ref{limitgcst} gives
\begin{align*}
U_s\left(\Psi_{\text{FQHE}}\right)=e^{-\sum_{i=1}^{N_e}\kappa_s(y_{s,j})}\exp\left(\frac{s}{2}\sum_{i=1}^{N_e}\frac{\partial^2}{\partial \theta_i^2}\right)\left(\prod_{1\leq i<j\leq N_e}(w_{s,j}-w_{s,j})^m\right)1^{\text{U}(1)}\otimes \sqrt{dz_s}.
\end{align*}

Given $\lambda=(\lambda_1,...,\lambda_{N_e})$ with $ 0\leq \lambda_1< \lambda_2 <....<\lambda_{N_e}\leq  N$, where $N=\deg L$, we have an associated Slater determinant state
\begin{align*}
\Psi_s^{\lambda}\defeq\sigma_s^{\lambda_1}\wedge....\wedge \sigma_s^{\lambda_{N_e}},
\end{align*} 
where $\{\sigma^{m}_s\}_{m=0}^{N}$ is again the monomial basis of $H^0(\mathbb{CP}^1,L_s)$ given in (\ref{eqq29}). The states $\{\Psi_s^{\lambda}\}_{0\leq \lambda_1<\lambda_2<...<\lambda_{N_e}\leq N}$ form a basis of $\Lambda^{N_e}H^{0}(\mathbb{CP}^1,L_s)$ and moreover,
\begin{align}\label{slaterus}
U_s\left(\Psi_{0}^{\lambda}\right)= e^{-\frac{s}{2}\sum_{i=1}^{N_e}\lambda_i^{2}}\Psi_{s}^{\lambda}.
\end{align}
Since the Laughlin states with odd filling fraction belong to $\Lambda^{N_e}H^{0}(\mathbb{CP}^1,L_s)$, we will be able to write
\begin{align*}
\Psi_{\text{Laughlin}}^{\nu=\frac{1}{m}}=\sum_{\lambda}a^{\nu}_{\lambda} \Psi^{\lambda}_0,
\end{align*} 
so that
\begin{align}\label{slaterlaughlinev}
U_s\left(\Psi_{\text{Laughlin}}^{\nu}\right)=\sum_{\lambda}  e^{-\frac{s}{2}\sum_{i=1}^{N_e}\lambda_i^{2}}a^{\nu}_{\lambda}\Psi_s^{\lambda}
\end{align}\par
Gerald V. Dunne in his paper~\cite{dunne:93} gives a method to determine these coefficients in a combinatorial manner. These become harder to compute the bigger $N_e$ is, and so we only compute exact density profiles for $N_e=2$ and $N_e=3$.\par
For $N_e=2$, one obtains
\begin{align*}
\Psi_{\text{Laughlin}}^{\nu=1/3}= \Psi_0^{(0,3)}-3\Psi_0^{(1,2)}
\end{align*}
and thus the evolved state is (cf. Eq.~\eqref{slaterlaughlinev})
\begin{align*}
U_s(\Psi_{\text{Laughlin}}^{\nu=1/3})=e^{-\frac{s}{2}9}\Psi_s^{(0,3)} -3e^{-\frac{s}{2}5}\Psi_s^{(1,2)}.
\end{align*}
For $N_e=3$,
\begin{align*}
\Psi_{\text{Laughlin}}^{\nu=1/3}=\Psi_{0}^{(0,3,6)}-3\Psi_{0}^{(1,2,6)}-3\Psi_{0}^{(0,4,5)}+6\Psi_{0}^{(1,3,5)}-15\Psi_0^{(2,3,4)},
\end{align*}
so that
\begin{align*}
U_s(\Psi_{\text{Laughlin}}^{\nu=1/3})=&e^{-\frac{s}{2} 45}\Psi_{s}^{(0,3,6)}-3e^{-\frac{s}{2} 41}\Psi_{s}^{(1,2,6)}-3e^{-\frac{s}{2}41}\Psi_{s}^{(0,4,5)}\\
&+6e^{-\frac{s}{2} 35}\Psi_{s}^{(1,3,5)}-15e^{-\frac{s}{2} 29}\Psi_s^{(2,3,4)}.
\end{align*}
\section{Evolution of density profiles by imaginary time flows on the sphere and on the plane}
\label{densityprofile}

\subsection{Evolution of density profiles }

We wish to evaluate the density profiles given by
\begin{equation}
\rho_s(x)\defeq\frac{\left< U_s\left(\Psi_{\text{Laughlin}}^{1/3}\right),\sum_{i=1}^{N_e}\delta(x,x_i)U_s\left(\Psi_{\text{Laughlin}}^{1/3}\right)\right>}{\norm{ U_s\left(\Psi_{\text{Laughlin}}^{1/3}\right)}^2_{L^2}}.
\end{equation}
with
\begin{align*}
U_s\left(\Psi_{\text{Laughlin}}^{1/3}\right)=\sum_{\lambda}a_{\lambda} e^{-\frac{s}{2}\sum_{i=1}^{N_e}\lambda_i^2} \Psi_s^{\lambda},
\end{align*}
We also have orthogonality of the Slater determinant states, i.e.,
\begin{align}\label{orthogonal}
\left<\Psi_s^{\lambda},\Psi_s^{\lambda'}\right> =\delta^{\lambda,\lambda'}f_s^{\lambda}
\end{align}
with
\begin{align*}
f_s^{\lambda}=\prod_{i=1}^{N_e}\norm{\sigma_s^{\lambda_{i}}}^2_{L^2},
\end{align*}
due to rotational symmetry (orthogonality of $e^{im\theta}$ and $e^{in\theta}$, $m\neq n$).
Therefore, we find that
\begin{align*}
\rho_s(x)=\frac{1}{\norm{ U_s\left(\Psi_{\text{Laughlin}}^{1/3}\right)}_{L^2}}\sum_{\lambda,\lambda'}\bar{a}_\lambda a_{\lambda'} e^{-\frac{s}{2}\sum_{i=1}^{N_e}\left(\lambda_i^2+\lambda_i^{'2}\right)} \left< \Psi_s^{\lambda}, \sum_{j=1}^{N_e}\delta(x,x_j) \Psi_s^{\lambda'}\right>,
\end{align*}
and since $\{\Psi_s^\lambda\}$ form an orthogonal basis (cf. Eq.~\eqref{orthogonal}), we have
\begin{align*}
\rho_s(x)&=\frac{1}{\norm{ U_s\left(\Psi_{\text{Laughlin}}^{1/3}\right)}_{L^2}}\sum_{\lambda}|a_\lambda^\nu|^2  e^{-s\sum_{i=1}^{N_e}\lambda_i^2} \left< \Psi_s^{\lambda}, \sum_{j=1}^{N_e}\delta(x,x_j) \Psi_s^{\lambda}\right>
\end{align*}
where
\begin{equation}\label{densityfraquinha}
\norm{U_s\left(\Psi_{\text{Laughlin}}^{1/3}\right)}_{L^2}^2= \sum_{\lambda}|a_\lambda^\nu|^2e^{-s\sum_{i=1}^{N_e}\lambda_i^2}\norm{\Psi^{\lambda}_s}_{L^2}=\sum_{\lambda}|a_\lambda^\nu|^2e^{-s\sum_{i=1}^{N_e}\lambda_i^2}\prod_{i=1}^{N_e}\norm{\sigma_s^{\lambda_i}}_{L^2}.
\end{equation}
Finally, we note that
\begin{align*}
    \left< \Psi_s^{\lambda}, \sum_{j=1}^{N_e}\delta(x,x_j) \Psi_s^{\lambda}\right>&=\sum_{j=1}^{N_e}\int_{\left(\C\mathbb{P}^1\right)^{N_e}}\delta(x-x_j)\prod_{i=1}^{N_e}\nu\left(\sigma_s^{\lambda_i}(x_i),\sigma_s^{\lambda_i}(x_i)\right)dx_i\\
    &=\sum_{j=1}^{N_e}\nu\left(\sigma_s^{\lambda_j}(x),\sigma_s^{\lambda_j}(x)\right)\int_{\left(\C\mathbb{P}^1\right)^{N_e-1}}\prod_{i=1,i\neq j}^{N_e}\nu\left(\sigma_s^{\lambda_i}(x_i),\sigma_s^{\lambda_i}(x_i)\right)dx_i\\
    &=\sum_{j=1}^{N_e}\nu\left(\sigma_s^{\lambda_j}(x),\sigma_s^{\lambda_j}(x)\right)\frac{\prod_{i=1}^{N_e}\norm{\sigma_s^{\lambda_i}}_{L^2}}{\norm{\sigma_s^{\lambda_j}}_{L^2}},
\end{align*}
where $\nu$ denotes the pointwise hermitian product. Therefore,  (\ref{densityfraquinha}) becomes
\begin{equation}\label{densityfortezinha}
    \rho_s(x)=\frac{\sum_{\lambda}|a_\lambda^\nu|^2  e^{-s\sum_{i=1}^{N_e}\lambda_i^2} \prod_{i=1}^{N}\norm{\sigma_s^{\lambda_i}}_{L^2}\sum_{j=1}^{N_e}\frac{\nu\left(\sigma_s^{\lambda_j}(x),\sigma_s^{\lambda_j}(x)\right)}{\norm{\sigma^{\lambda_j}_s}_{L^2}}}{\sum_{\lambda}|a_\lambda^\nu|^2e^{-s\sum_{i=1}^{N_e}\lambda_i^2}\prod_{i=1}^{N_e}\norm{\sigma_s^{\lambda_i}}_{L^2}}.
\end{equation}
We can still make the result a bit more explicit by using
\begin{align*}
\nu\left(\sigma^{\lambda_j}_s(x),\sigma^{\lambda_j}_s(x)\right)=e^{-2\kappa_s}\norm{dz_s} \; |w_s|^{2\lambda_j}
\end{align*}
and thus obtaining
\begin{equation}
    \rho_s(x)=e^{-2\kappa_s}\norm{dz_s}\frac{\sum_{\lambda}|a_\lambda^\nu|^2  e^{-s\sum_{i=1}^{N_e}\lambda_i^2} \prod_{i=1}^{N}\norm{\sigma_s^{\lambda_i}}_{L^2}\sum_{j=1}^{N_e}\frac{|w_s|^{2\lambda_j}}{\norm{\sigma^{\lambda_j}_s}_{L^2}}}{\sum_{\lambda}|a_\lambda^\nu|^2e^{-s\sum_{i=1}^{N_e}\lambda_i^2}\prod_{i=1}^{N_e}\norm{\sigma_s^{\lambda_i}}_{L^2}}.
\end{equation}

\subsection{Large $s$ asymptotics}
We will now consider the large $s$ limit of $\rho_s(x)$. From (\ref{quantumcalc}) and (\ref{finallimit}), we get, as $s\to\infty$
\begin{equation}
    \nu\left(\sigma_s^{\lambda_i},\sigma_s^{\lambda_i}\right)\sim\sqrt{\pi}e^{s|\lambda_i|^2}e^{2g_p(\lambda_i)}\delta_{\lambda_i},
\end{equation}
and
\begin{align}
\norm{\sigma_s^{\lambda_i}}_{L^2}^2\sim\sqrt{\pi}e^{s|\lambda_i|^2}e^{2g_p(\lambda_i)},
\end{align}
Using both to take the limit $s\to\infty$ in (\ref{densityfortezinha}), we obtain
\begin{equation}\label{limitdensityfortezinha}
    \lim_{x\to\infty}\rho_s(x)=\frac{\sum_{\lambda}|a_\lambda^\nu|^2 e^{2g_p(\lambda_i)}\sum_{j=1}^{N_e}\delta_{\lambda_j}}{\sum_{\lambda}|a_\lambda^\nu|^2e^{2g_p(\lambda_i)}}.
\end{equation}

In particular, for the IQHE state, we have $N_e=N$ with $\lambda=(0,1,...,N)$ being the only possible state, so we immediately obtain
\begin{equation}
    \lim_{x\to\infty}\rho_s(x)=\sum_{j=1}^N\delta_{\lambda_j}
\end{equation}
corresponding to a uniform distribution of Bohr-Sommerfeld leaves on $\mathbb{CP}^1$, in particular those corresponding to the integer points $P\cap\mathbb{Z}=[-1/2,N-1/2]\cap\mathbb{Z}=\{0,...,N-1\}=[0,N-1]\cap\mathbb{Z}$.
\subsection{Examples with $\nu=1/3$ filling}\label{finaaaal}
We will now analyze (\ref{limitdensityfortezinha}) in the specific case where $\nu=1/3$.\par
First, we look at examples with few particles, for which the combinatoric coefficients $a_\lambda^\nu$ are easily determined (see Ref.~\cite{dunne:93}). We can then plot the density profiles of these states for different values of $s$, as is done in Figures~\ref{figura1} and~\ref{figura2} for the sphere and in Figures~\ref{figura3} and 
Figures~\ref{figura4} for the plane, for systems of $2$ and $3$ particles, respectively.\par
The first thing we observe is that the ratios between heights of different peaks indeed seem to approach well-defined limits, very close to
\begin{equation}\label{limitratio}
    R_{m,n}=\frac{\sum_{\lambda\colon m\in\lambda}\,\lvert a^\nu_\lambda\rvert^2 e^{2\sum_{i=1}^{N_e}g_P(\lambda_i)}}{\sum_{\Tilde{\lambda}\colon n\in\Tilde{\lambda}}\,\lvert a^\nu_{\Tilde{\lambda}}\rvert^2 e^{2\sum_{i=1}^{N_e}g_P(\Tilde{\lambda}_i)}}.
\end{equation}
between peaks at $x=m$ and $x=n$ corresponding to the density profile (\ref{limitdensityfortezinha}), which yields
\begin{align*}
    N_e&=2\implies R_{0,1}\approx 1.08\\[8pt]
    N_e&=3\implies\left\{\begin{array}{ll}
    R_{0,1}\approx 1.03\\[4pt]
    R_{1,2}\approx 1.01
    \end{array}
    \right.
\end{align*}
for the sphere and
\begin{align*}
    N_e&=2\implies R_{0,1}\approx 0.35\\[8pt]
    N_e&=3\implies\left\{\begin{array}{ll}
    R_{0,1}\approx 0.81\\[4pt]
    R_{1,2}\approx 0.50
    \end{array}
    \right.
\end{align*}
for the plane.\par
In Figures~\ref{figura7} and~\ref{figura8} we show imaginary time Hamiltonian flow evolutions using only the prequantum evolution operator (\ref{eq: preevolutionop}), for $2$ electrons on the plane and $3$ on the sphere. We see that, as expected from the results found in Proposition~\ref{limitgcst}, the Laughlin states converge to a single Slater determinant, namely the one with the largest $\abs{\lambda}^2$. This is purely a consequence of non-unitarity of the evolution operator used.\par

We now consider the large $N_e$ limit. We see from (\ref{limitdensityfortezinha}) that how much a certain state $\lambda$ contributes to the height of its peaks depends on the coefficient $a^{\nu}_\lambda$ and the factor
$$S(\lambda)=\exp\left(2\sum_{i=1}^{N_e}g_P(\lambda_i)\right).$$\par
For large $N_e$, the dominant $|a^{\nu}_\lambda|^2$ comes from the maximally ``bunched'' state $\lambda^M=(0,1,...,N_e-1) +(N_e-1,N_e-1,...,N_e-1)$, whose coefficient satisfies $|a^{1/3}_{\lambda^M}|=(2N_e-1)!!$ (see Ref.~\cite{dunne:93}). In contrast, the minimum $|a^{\nu}_\lambda|^2$ comes from the most uniform state $\lambda^m=3(0,1,...,N_e-1)$, for which $|a^{1/3}_{\lambda^m}|=1$. By plotting the ratio of the factors $|a^{1/3}_{\lambda}|^2S(\lambda)$ for these two states as a function of $N_e$ (Figure \ref{figura5}), we see that the change in the $S$ factor greatly outweighs the change in the combinatorial factor $|a^{1/3}_\lambda|^2$.\par
Since the function $g_P$ has a minimum in the center of the polytope and increases away from it, the state $\lambda^M$ is the one for which $S$ is the smallest. However there are states with higher concentration in the poles that have much larger $S$ than the value of $S(\lambda^m)$. As this is the leading factor in determining the heights of peaks, we expect these states to have much more relevance and thus the density to be higher near the poles than in the center.\par
A similar analysis can be made for many particle states on the plane, since the combinatorial factors are the same and do not change significantly for different Slater determinants in comparison with the function
$$\Tilde{S}(\lambda)=\exp\left(2\sum_{i=1}^{N_e}\Tilde{g}_P(\lambda_i)\right),$$
as can be seen in the plot in Figure~\ref{figura6}.

\newpage

\begin{figure}
    \centering
    \includegraphics[width=\linewidth]{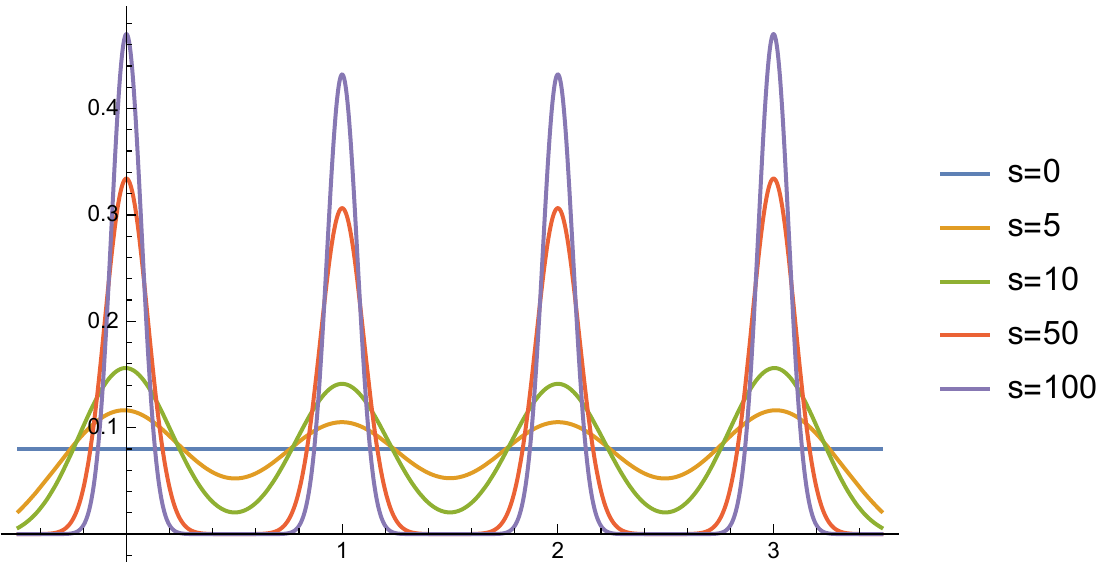}
    \caption{Density profiles of the evolved states of $2$ particles on the sphere with $s=0$, $s=5$, $s=10$, $s=50$ and $s=100$.}
    \label{figura1}
\end{figure}
\begin{figure}
    \centering
    \includegraphics[width=\linewidth]{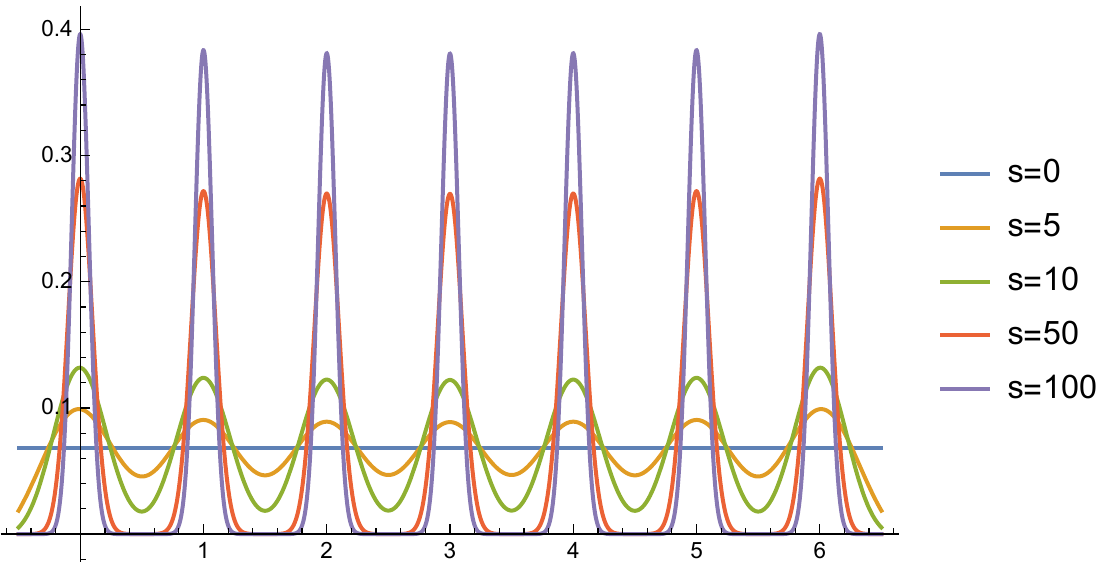}
    \caption{Density profiles of the evolved states of $3$ particles on the sphere with $s=0$, $s=5$, $s=10$, $s=50$ and $s=100$.}
    \label{figura2}
\end{figure}
\begin{figure}
    \centering
    \includegraphics[width=\linewidth]{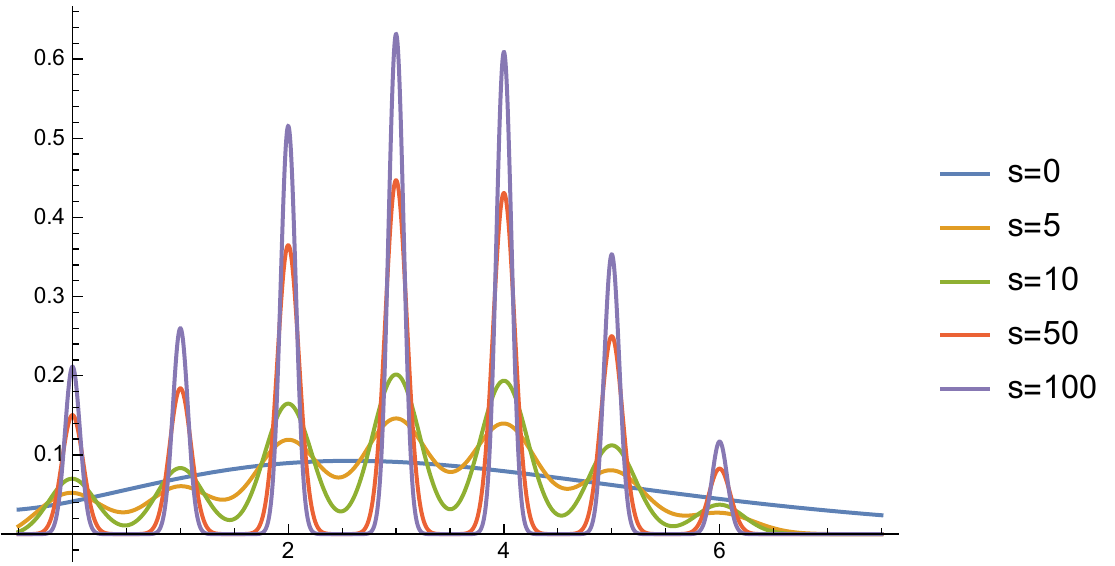}
    \caption{Density profiles of the evolved states of $2$ particles on the plane with $s=0$, $s=5$, $s=10$, $s=50$ and $s=100$.}
    \label{figura3}
\end{figure}
\begin{figure}
    \centering
    \includegraphics[width=\linewidth]{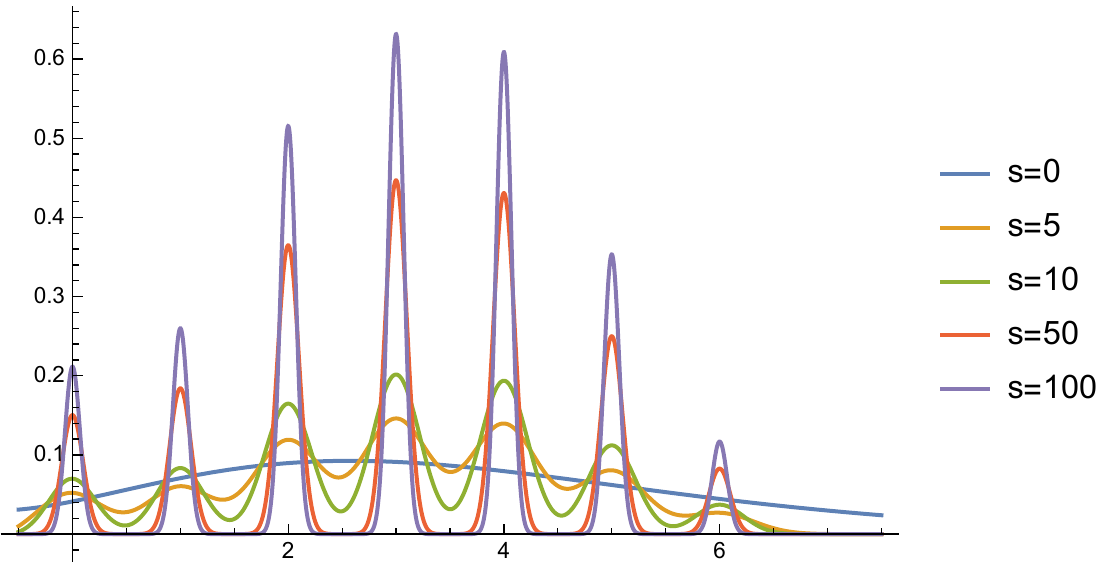}
    \caption{Density profiles of the evolved states of $3$ particles on the plane with $s=0$, $s=5$, $s=10$, $s=50$ and $s=100$.}
    \label{figura4}
\end{figure}
\begin{figure}
    \centering
    \includegraphics[width=\linewidth]{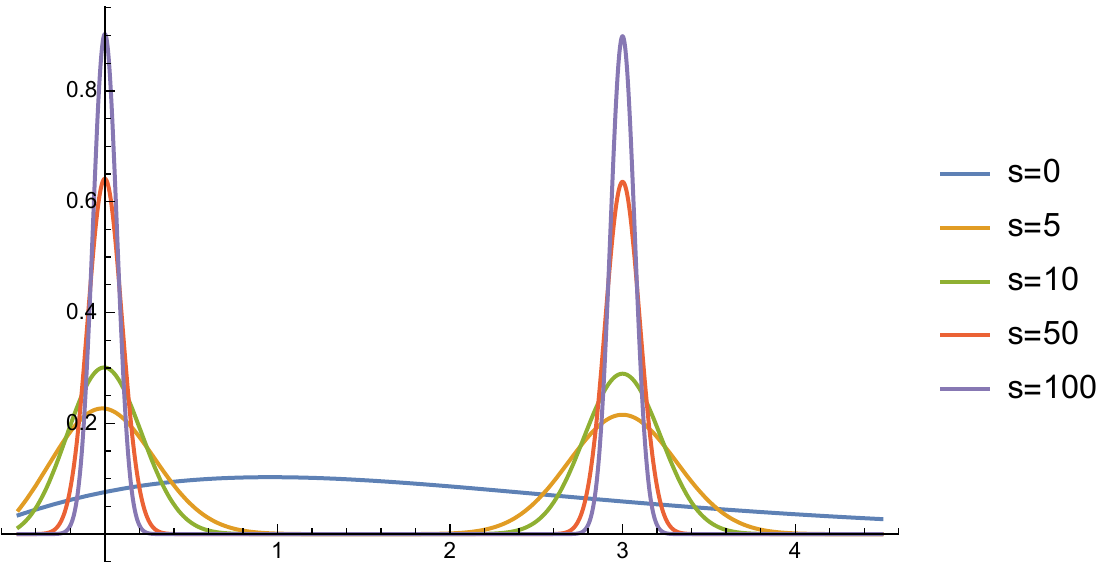}
    \caption{Density profiles of states of $2$ particles on the plane evolved only with prequantum operator for $s=0$, $s=5$, $s=10$, $s=50$ and $s=100$.}
    \label{figura7}
\end{figure}
\begin{figure}
    \centering
    \includegraphics[width=\linewidth]{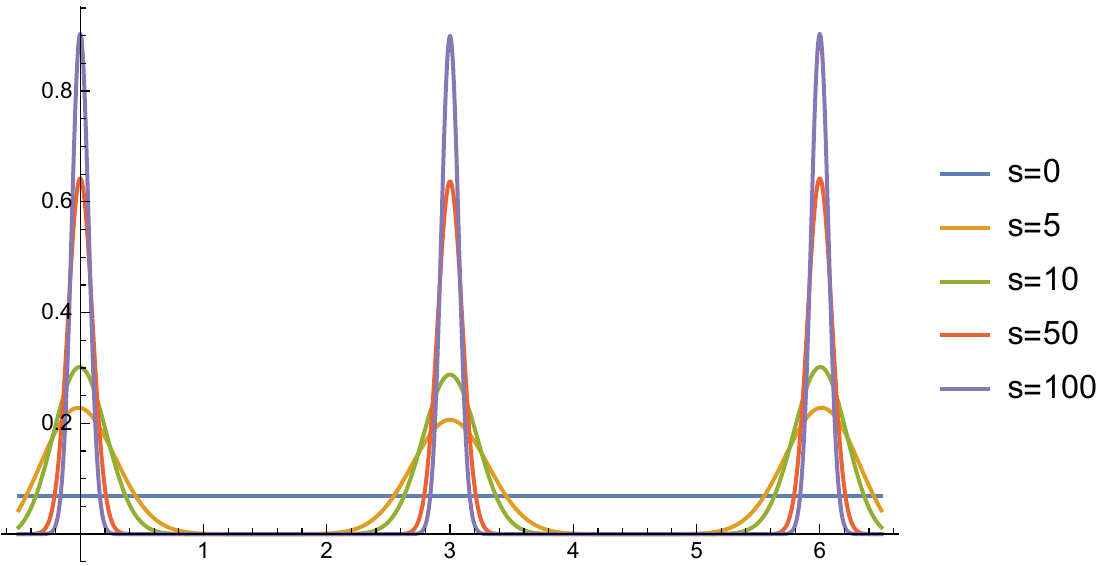}
    \caption{Density profiles of states of $3$ particles on the sphere evolved only with prequantum operator for $s=0$, $s=5$, $s=10$, $s=50$ and $s=100$.}
    \label{figura8}
\end{figure}
\begin{figure}
    \centering
    \includegraphics[width=0.8\linewidth]{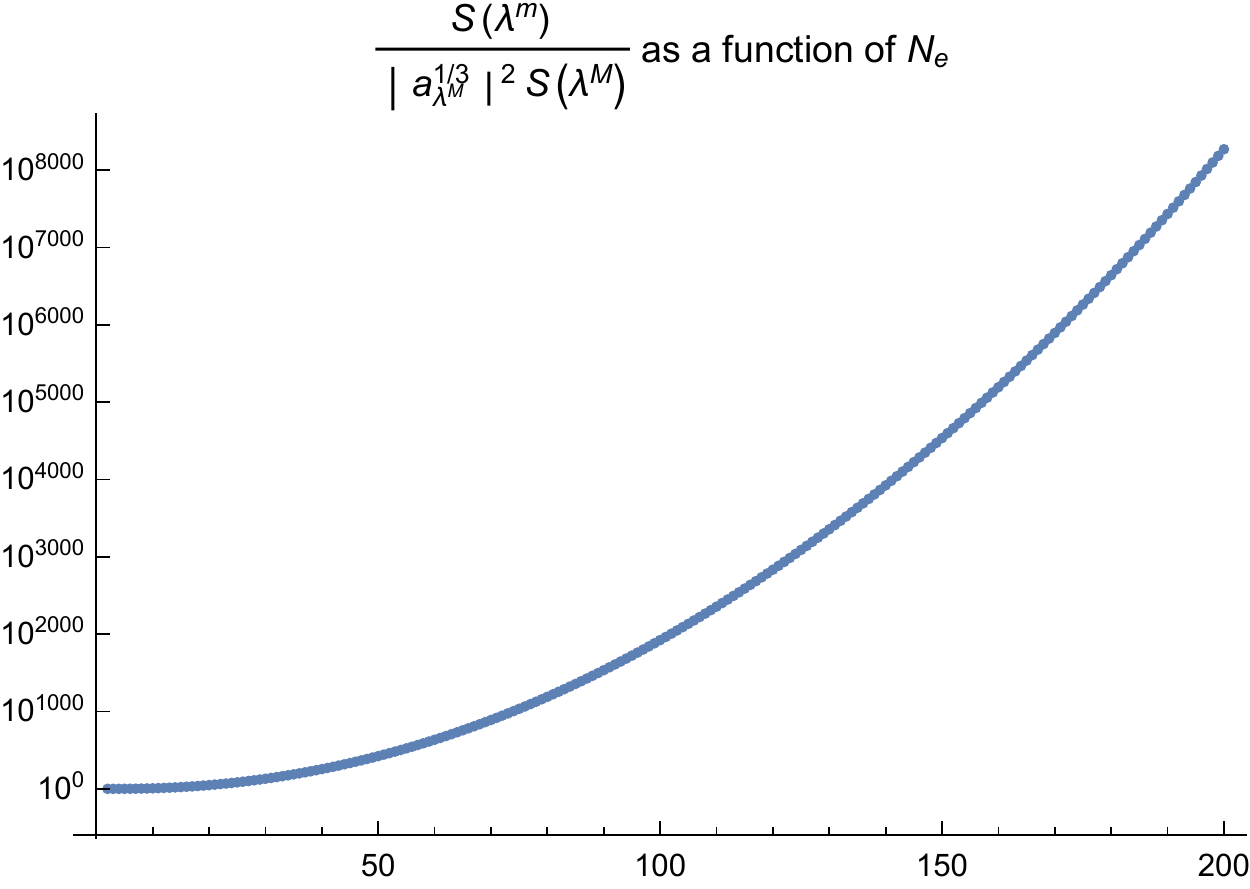}
    \caption{}
    \label{figura5}
\end{figure}
\begin{figure}
    \centering
    \includegraphics[width=0.8\linewidth]{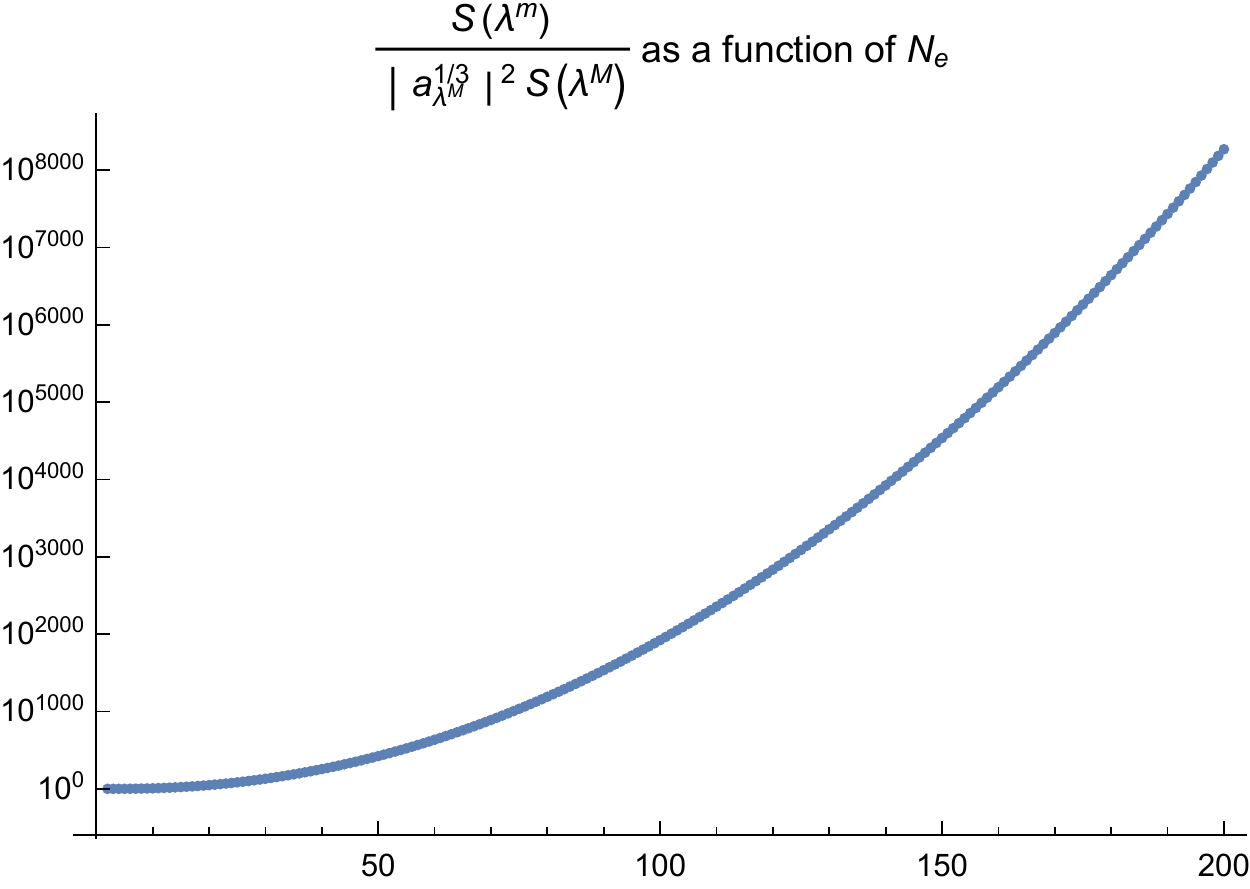}
    \caption{}
    \label{figura6}
\end{figure}


\vspace{5cm}
\section{Conclusion}
\label{chap:conclusion}
With this work, we achieved our goal of using the GCST to obtain a detailed description of the evolution of Laughlin states on the sphere under deformations of the geometry.\par
Starting with one particle states, the results of Proposition~\ref{limitgcst}, not only describe how the GCST acts on the quantum Hilbert space of the sphere, but also allow us compare this evolution with the one obtained via the prequantum evolution operator defined by (\ref{eq: preevolutionop}). We see that, although the latter gives a very intuitive evolution, yielding states considered in the literature to be Laughlin states (see \cite{klevtsov:19}), it is highly non-unitary, which is an issue that we observe to be precisely fixed by considering the GCST instead.\par
The results of chapter \ref{chap:implement} then highlight the consequences of the aforementioned proposition to many particle states, in particular Laughlin states. The density profile obtained in expression (\ref{densityfortezinha}) provides a very useful description of the evolution of Laughlin states, and cases with few particles are computed in detail (see Figures~\ref{figura1},~\ref{figura2},~\ref{figura3} and~\ref{figura4}) for the sphere and the plane.\par
We pay special attention to extreme deformations ($s\to\infty$), where we see the density converge to integer points on the polytope with peak ratios given by well-defined limits (cf. Eq.~\eqref{limitratio}), dependent purely on physical properties of the system. This heavily contrasts with the evolution obtained by only considering the prequantum operator, where non-unitarity causes convergence of the Laughlin states to specific Slater determinants (see Figures~\ref{figura7} and~\ref{figura8}).\par
Finally, we see that the density profiles we obtained also provide information related to systems with a large number of particles as $s\to\infty$, namely higher concentrations near the poles for the sphere and near the center for the plane.\par
Further work would be useful to obtain a systematic and reliable way to compute the coefficients $a_\lambda$ present in (\ref{densityfortezinha}). This would allow for an exact computation of the evolution of states for an arbitrary number of particles as $s\to\infty$. The study of analogous deformations to those considered here, namely on the torus for which a Laughlin state was also given by Haldane (see \cite{haldane:83}), or induced by different Hamiltonians would also be extremely relevant.

\section*{Acknowledgements}
GM thanks the Calouste Gulbenkian Foundation for a fellowship
Novos Talentos em Tecnologias Quânticas and CAMGSD
for a BIL fellowship in the beginning of this work.
BM acknowledges the support from 
SQIG,  Instituto de Telecomunica\c{c}\~oes, 
and  FCT/Portugal through the projects 
UIDB/50008/2020, H2020 project SPARTA, as well as projects QuantMining POCI-01-0145-FEDER-031826 and PREDICT PTDC/CCI-CIF/29877/2017.
JM and JPN were funded by 
 FCT/Portugal through
the projects CAMGSD UIDB/04459/2020 and  PTDC/MAT-OUT/28784/2017. PM 
thanks CAMGSD
for a BIL fellowship in the beginning of this work and
 the
ERC-SyG project “Recursive and Exact New Quantum Theory” (ReNewQuantum).

\bibliographystyle{unsrt}
\bibliography{M4N}
\end{document}